\documentclass[final,3p,times]{elsarticle}

\makeatletter
\def\ps@pprintTitle{%
 \let\@oddhead\@empty
 \let\@evenhead\@empty
 \def\@oddfoot{\centerline{\thepage}}%
 \let\@evenfoot\@oddfoot}
\makeatother

\usepackage{graphicx}
\usepackage{amssymb,amsmath}
\usepackage{amsthm}
\usepackage{bm}
\usepackage{cancel}

\newtheorem{theorem}{Theorem}

\newtheorem{lemma}{Lemma}

\DeclareMathOperator{\sech}{sech}

\newcommand{\ie}{{\it i.e.}}
\newcommand{\p}{\partial}

\newcommand{\ex}{\bm{\hat{e}}_1}
\newcommand{\ey}{\bm{\hat{e}}_2}
\newcommand{\ez}{\bm{\hat{e}}_3}

\newcommand{\ev}{{\bf\hat{e}_\varphi}}
\renewcommand{\eth}{{\bf\hat{e}_\theta}}

\newcommand{\thth}{{\bf \hat{e}_{\theta,\theta}}}
\newcommand{\thv}{{\bf \hat{e}_{\theta,\varphi}}}
\newcommand{\vth}{{\bf \hat{e}_{\varphi,\theta}}}
\newcommand{\vv}{{\bf \hat{e}_{\varphi,\varphi}}}

\newcommand{\DM}{D}
\newcommand{\dm}{\lambda}
\newcommand{\Anisotropy}{K}
\newcommand{\anisotropy}{k}

\newcommand{\magn}{\bm{m}}
\newcommand{\heff}{\bm{f}}
\newcommand{\lex}{\ell_{\rm ex}}
\newcommand{\ldm}{\ell_{\rm D}}
\newcommand{\ldw}{\ell_{\rm w}}
\newcommand{\Energy}{E}

\newcommand{\linmom}{P}
\newcommand{\pold}{p}

\newcommand{\tmagn}{\mathcal{M}}
\newcommand{\vel}{c}
\newcommand{\velmax}{\vel_\ell}
\newcommand{\vabs}{v}

\newcommand{\loverk}{\beta}
\newcommand{\lifshitz}{\mathcal{L}}

\begin{document}

\begin{frontmatter}


\title{Traveling domain walls in chiral ferromagnets}


\author{Stavros Komineas}
\address{Department of Mathematics and Applied Mathematics, University of Crete, 71305 Heraklion, Greece}
\author{Christof Melcher}
\address{Department of Mathematics I \& JARA Fundamentals of Future Information Technology, RWTH Aachen University, 52056 Aachen, Germany}
\author{Stephanos Venakides}
\address{Department of Mathematics, Duke University, Durham, NC, USA}

\begin{abstract}
We show that chiral symmetry breaking enables traveling domain wall solution for the conservative Landau-Lifshitz equation of a uniaxial ferromagnet with Dzyaloshinskii-Moriya interaction. In contrast to related domain wall models including stray-field based anisotropy, traveling wave solutions are not found in closed form. For the construction we follow a topological approach and provide details of solutions by means of numerical calculations.
\end{abstract}


\end{frontmatter}

\section{Introduction}
\label{sec:introduction}

Magnetic domain walls (DW) are transition layers separating domains of different magnetizations in magnetic materials.
They are fundamental for understanding domain structure.
In mathematical idealization, they are a special form of kinks, connecting different asymptotic equilibrium states on the unit sphere.
Static domain walls occur as stable equilibria of the micromagnetic energy, that depends on the material crystal structure and sample geometry.
The simplest example arises from the combination of exchange and uniaxial anisotropy.
These interactions are symmetric with respect to rotations around the anisotropy axis, giving rise to a one-parameter group of degenerate static domain wall configurations in the form of geodesic connections (meridians) of the two antipodes of the magnetization (poles) $m_3=\pm 1$.

Antisymmetric exchange, also called the Dzyaloshinskii-Moriya interaction (DMI) \cite{Dzyaloshinskii_JETP1957,Moriya_PhysRev1960}, is present in a class of ferromagnetic materials whose crystal structure lacks inversion symmetry.
The DMI has profound consequences for the equilibrium domain structure \cite{Dzyaloshinskii_JETP1964,BogdanovHubert_JMMM1994}.
It breaks the symmetry of the interaction energy and the degeneracy of the solution space.
In a DM material, if anisotropy is strong enough, the fully aligned ferromagnetic state is the ground state and a domain wall is an excitation.
In the presence of DMI two specific meridians are selected, on which domain walls, stable and unstable, are maintained.
The walls are of Bloch type, i.e., the magnetization vector is at the angle $\varphi= \pm \pi/2$ to the direction of propagation.

The dynamics of domain walls is particularly interesting because a moving domain wall corresponds to a varying domain structure.
Also, a domain wall can play the role of a unit of information that is transmitted via wall propagation.
In general, domain wall dynamics is governed by the (conservative) Landau-Lifshitz (LL) \cite{LandauLifshitz_PZS1935} or the (dissipative) Landau-Lifshitz-Gilbert (LLG) equation \cite{ODell}, derived from an energy functional.
In a magnet with symmetric Heisenberg exchange and uniaxial anisotropy alone, traveling domain wall solutions for the corresponding LL equation are not possible.
A traveling domain wall solution of the LL equation, called the Walker solution \cite{SchryerWalker_JAP1974,ODell}, is obtained in a model with an extra anisotropy, stemming from magnetostatic stray-field interaction, that is breaking the symmetry around the original uniaxial anisotropy axis. 
These solutions are typically discussed in a model with damping and external field, but they are actually exact solutions of the conservative model (see \ref{sec:WalkerWall}).
In axisymmetric systems, DW solutions moving with constant velocity but with precessional oscillations in time have been discussed e.g. in \cite{GoussevRobbinsSlastikov2010, MelcherRademacher2017}.

A common mathematical feature of the models which have been shown to support traveling domain walls (without temporal oscillations) is a breaking of rotational symmetry through stray fields, see e.g. \cite{MelcherDWmotion} and \cite{CapellaMelcherOtto} for the extreme case of a traveling N\'eel wall.
We notice that this effect can also be achieved by chiral symmetry breaking.
Therefore, we are motivated to explore the possibility of traveling domain walls in the LL equation in the presence of DMI.

The possibility for freely (unforced) traveling domain walls in chiral magnets is indicated by numerical and analytical studies for the wall mobility, as a response to an applied magnetic field.
It has been shown that this is enhanced by the DMI compared to the standard Walker domain wall \cite{ThiavilleRohart_EPL2012}.
The increased mobility for chiral DW is also manifested in the case of motion due to an applied electrical current \cite{TretiakovAbanov_PRL2010,EmoriBauerBeach_nmat2013,Brataas_nnano2013,GoussevRobbinsSlastikovTretiakov}.
The dynamics of the DW has been discussed largely within collective coordinate approaches, called the $q-\Phi$ model, and numerical simulations \cite{ThiavilleRohart_EPL2012,Wieser_PSS2015}.

We study the effect of chiral symmetry breaking on the laws of dynamics for magnetic domain walls.
The dynamics of the magnetization can be linked to the symmetries of the magnetic interactions via associated conservation laws.
We show that the symmetry-breaking introduced by the DMI in a film with perpendicular anisotropy allows for propagating DW solutions within the Landau-Lifshitz equation, even if the magnetostatic field is negligible.

At the stable meridian ($\varphi=\pi/2$) the wall is static and deviation from this at the symmetry point (center of wall) gives rise to a tilting angle $\Delta \varphi$.
Tilting, however, induces a dynamic response.
Here, we prove the existence of propagating domain wall solutions in the form of traveling waves in the conservative model by employing a topological argument.

\begin{theorem}
\label{thm:summary}
Let the easy-axis anisotropy parameter $\anisotropy^2>0$ and the DMI parameter $\dm/\anisotropy > 0$ be sufficiently small.
Then, for sufficiently small tilting angle, there exists a traveling domain wall.
\end{theorem}

The existence and stability of traveling domain walls in the presence of damping and an external driving magnetic field can also be proven by means of a perturbative approach \cite{KomineasMelcherVenakides_preparation2018}.

Details of the domain wall profiles come as a product of the proof of their existence.
The propagating DMI walls have a profile that is substantially more complicated than the Walker DWs.
For the case of bulk DMI a Bloch wall is static but a N\'eel wall is propagating with maximum velocity in stark contrast to the Walker propagating wall.
We calculate numerically the profiles of propagating domain walls for velocities $0 < \vel < \velmax$ where $\velmax$ is the maximum achieved velocity.

\section{The Landau-Lifshitz equation}
\label{sec:LL}

We consider a ferromagnet described by the magnetization vector $\bm{M}=\bm{M}(x,t)$ that is a function of space and time but it has a constant magnitude $|\bm{M}|=M_s$, where $M_s$ is called the saturation magnetization.
Statics and dynamics of the magnetization are governed by the Landau-Lifshitz equation, whose dimensionless form reads 
\begin{equation}  \label{eq:LL}
 \p_t\magn = -\magn\times\heff
\end{equation}
where $\magn=\bm{M}/M_s$ is the normalized magnetization, in component form $\magn=(m_1,m_2,m_3)$.
The variable $t$ is the dimensionless time that is measured in units of $t_0 = 1/(\gamma_0 \mu_0 M_s)$, where $\gamma_0$ is the gyromagnetic ratio and $\mu_0$ the permeability of vacuum.
The effective field $\bm{f}$ contains the interactions in the material.
We will assume a ferromagnet with exchange, an easy-axis anisotropy, and a DMI.
We will study configurations where the magnetization is varying in only one space direction, that is, we assume $\magn=\magn(x,t)$.
The energy of such a system is
\begin{equation}  \label{eq:energy}
\Energy(\magn)  = \int \left[ \frac{(\p_x\magn)^2}{2} 
+ \frac{\anisotropy^2}{2} (1-m_3^2) 
+ \dm (\magn\times\p_x\magn)\cdot\ex \right]\,dx
\end{equation}
where $\ex$ is the unit vector for the magnetization in the $x$ direction.
We measure distance in units of exchange length $\lex = \sqrt{2A/(\mu_0 M_s^2)}$ where $A$ is the exchange constant.
There are two length scales in the model, $\ldw = \sqrt{A/K}$, where $\Anisotropy$ is the anisotropy constant, and $\ldm = 2A/|D|$, where $\DM$ is the DMI constant.
The dimensionless parameters appearing in the energy \eqref{eq:energy} are $\anisotropy^2 = 2 \Anisotropy/(\mu_0 M_s^2) = (\lex/\ldw)^2$, and $\dm = \lex/\ldm$.
We will consider $\dm > 0$; the case $\dm<0$ corresponds to the transformation $x\to -x$.
The general form of the DM term is given in terms of Lifshitz invariants
\begin{equation} \label{eq:LifshitzInvariants}
\lifshitz_{jk} = (\magn\times\p_j\magn)_k.
\end{equation}
In the energy ~\eqref{eq:energy} we have only kept the Lifshitz invariant $\lifshitz_{11}$ in the $\ex$ direction corresponding to cubic DMI given by $\magn \cdot ( \nabla \times \magn)$. Replacing 
$\lifshitz_{11}$ by $\lifshitz_{12}$ (interfacial DMI)
or a linear combination of both yields a model that is mathematically equivalent modulo a rigid rotation around the $\ez$ axis.

The effective field entering \eqref{eq:LL} is obtained by varying the energy,
\begin{equation}  \label{eq:heff}
 \heff = -\frac{\delta \Energy}{\delta \magn} = \p_x^2\magn + \anisotropy^2\, m_3\ez - 2\dm\, \ex \times \p_x\magn.
\end{equation}

The uniform (ferromagnetic) states $\magn=(0,0\pm 1)$ are the simplest time-independent (static) solutions  of the LL equation \eqref{eq:LL}.
For large anisotropy, such that
\begin{equation}  \label{eq:conditionSpiral}
\anisotropy > \anisotropy_c \equiv \frac{\pi}{2} \dm,
\end{equation}
the ferromagnetic is the ground state of the system,
while for $\anisotropy < \anisotropy_c$ a {\it spiral configuration} becomes the ground state \cite{BogdanovHubert_JMMM1994}.
The period of the spiral increases for increasing anisotropy and goes to infinity as $\anisotropy \to \anisotropy_c$.

In this work we assume a material with $\anisotropy > \anisotropy_c$ and we are looking for domain wall solutions as excitations of the ferromagnetic ground state \cite{RohartThiaville_PRB2013,MuratovSlastikov_PRSA2016}.
A standard Bloch wall
\begin{equation}  \label{eq:BlochWall}
m_1 = 0,\qquad 
m_2 = \pm\sech(\anisotropy x),\qquad 
m_3 = \pm\tanh(\anisotropy x),
\end{equation}
for any combination of the signs,
is a solution of Eq.~\eqref{eq:LL} also in the presence of DMI ($\dm \neq 0$) as the contribution of the DM term on the right-hand-side of Eq.~\eqref{eq:LL} vanishes identically for these configurations.
As the DMI is chiral, the walls with the same signs for $m_2, m_3$ in Eq.~\eqref{eq:BlochWall} have lower energy, for $\dm > 0$.
The walls with opposite signs for $m_2, m_3$ are energy maxima.

One can easily prove that traveling DWs are not possible in model \eqref{eq:LL} when the DMI is not included in the effective field \eqref{eq:heff}.
This is one of the results obtained in \ref{sec:nonchiralMagnet} where a direct and complete solution for the domain walls of the system for $\dm=0$ is given.
For a more intuitive proof let us consider the total magnetization in the direction perpendicular to the film
\begin{equation}  \label{eq:tmagn}
\tmagn = \int_{-\infty}^\infty m_3\, dx
\end{equation}
in the sense of the Cauchy principle value
and calculate its time derivative using Eq.~\eqref{eq:LL}
\begin{equation}  \label{eq:tmagnDerivative}
\frac{d\tmagn}{dt} = -2\dm \int_{-\infty}^\infty m_1\p_x m_3\, dx.
\end{equation}
This result is a reflection of the fact that the exchange and anisotropy interactions are invariant with respect to rotations around the third axis of the magnetization, and therefore the total magnetization $\tmagn$ is conserved in the absence of DMI ($\dm=0$).
Since a propagating DW configuration is equivalent to expanding one domain (say, the ``up'' domain) in favor of the other (``down'' domain), thus changing $\tmagn$,
DW propagation is not possible in a model where the total magnetization $\tmagn$ is conserved.

In the model with effective field \eqref{eq:heff} it is entirely due to the DMI that the symmetry is broken and the associated conservation law is not valid, thus allowing for the possibility of propagating domain walls.
If we assume a rigid wall connecting the south pole ($m_3=-1$) at $x\to-\infty$ to the north pole ($m_3=1$) at $x\to\infty$ and traveling with velocity $\vel$ then we obtain from Eq.~\eqref{eq:tmagn} $\vel=-(1/2)d\tmagn/dt$ and from \eqref{eq:tmagnDerivative}
\begin{equation}  \label{eq:speedFromMagnetization}
\vel = \dm \int_{-\infty}^\infty m_1 \p_x m_3\, dx.
\end{equation}
This gives an upper bound for the speed
\begin{equation}  \label{eq:speedBound1}
|\vel| \leq 2 \dm.
\end{equation}
More generally, a Lifshitz invariant $\lifshitz_{1k}$ gives rise to an integrand $m_k \p_x m_3$ in Eq. \eqref{eq:speedFromMagnetization}. In particular, no non-trivial traveling DW solution is possible in the case $k=3$ corresponding to a wire along $\ez$ with cubic DMI and stray-field induced anisotropy.

Lastly, considering the conditions of Eq.~\eqref{eq:conditionSpiral} and Eq.~\eqref{eq:speedBound1}, we have, for positive $\vel$, the ordering 
\begin{equation}  \label{eq:parameter_ordering}
 0 < \frac{\vel}{2} < \dm \le \frac{2}{\pi}k. 
\end{equation}

\section{Derivation of a dynamical system for traveling waves}
\label{sec:dynamicalSystem}

Let us assume a rigid domain wall configuration propagating with a constant velocity $\vel$.
We substitute the traveling wave ansatz, $\magn = \magn(x-\vel t)$, in Eq.~\eqref{eq:LL} and this reduces to
\begin{equation}  \label{eq:LLtraveling}
\vel\magn' = \magn\times\heff.
\end{equation}
Our aim is to prove the existence of domain wall solutions for Eq.~\eqref{eq:LLtraveling} and to understand in detail their profiles.
We begin by writing explicitly Eq.~\eqref{eq:LLtraveling}
\begin{equation}  \label{eq:LLtravelingExplicit}
 \vel\magn'=\magn\times \magn''+k^2\magn\times m_3\ez - 2\dm\magn\times(\ex\times\magn')
\end{equation}
where we have the traveling wave $\magn=\magn(\xi),\; \xi = x-\vel t$ and the prime denotes differentiation with respect to $\xi$.

We define the orthonormal system of the three unit vectors
\begin{equation}  \label{eq:coordinateSystem}
\magn,\quad\frac{1}{\vabs}\magn', \quad \magn_\perp=\frac{1}{\vabs}\magn\times \magn',
\end{equation}
where we have defined $\vabs=|\magn'|$.
Both sides of Eq.~\eqref{eq:LLtravelingExplicit} are orthogonal to the magnetization vector $\magn$, hence, they lie on the tangent plane of the unit sphere at point $\magn$.
Thus, Eq.~\eqref{eq:LLtravelingExplicit} produces at most two independent scalar equations. We obtain these equations  by projecting Eq.~\eqref{eq:LLtravelingExplicit} on the two unit tangent vectors $\magn'/\vabs$ and $\magn_\perp$,

{\it Projection of Eq.~\eqref{eq:LLtravelingExplicit} on the vector $\magn_\perp$.}
We obtain
\begin{equation}  \label{eq:projection_mperp}
0 = \magn_\perp\cdot (\magn\times\magn'') + k^2 \magn_\perp\cdot(\magn\times\ez) \cos\theta.
\end{equation}
The DM term does not contribute as this is proportional to $\magn'$.
For the first term we have
\begin{equation}
\magn_\perp\cdot (\magn\times\magn'') 
= \magn_\perp\cdot (\magn\times\magn')'
= \frac{1}{2\vabs}\left(|\vabs\magn_\perp|^2|\right)' = \vabs'.
\end{equation}
For the second term we have
\begin{equation}
\magn_\perp\cdot(\magn\times\ez) = \frac{m_3'}{\vabs}.
\end{equation}
Inserting these results in Eq.~\eqref{eq:projection_mperp}, we obtain,
\begin{equation*}
\vabs' + k^2 \frac{m_3 m_3'}{\vabs} = 0.
\end{equation*}
which integrates to
\begin{equation}  \label{eq:conservationLaw0}
 \vabs^2+k^2 m_3^2 = k^2,
\end{equation}
where the choice of the integration constant arises from the requirement that the north and south poles $\magn=(0,0,\pm1)$ constitute constant equilibrium solutions.

We will make use in the following of the spherical parametrization for the magnetization vector
\begin{equation}
m_1 = \sin\theta \cos\varphi,\quad
m_2 = \sin\theta \sin\varphi,\quad
m_3 = \cos\theta.
\end{equation}
Eq.~\eqref{eq:conservationLaw0} reads
\begin{equation}  \label{eq:conservationLaw}
\vabs=k\sin\theta.
\end{equation}
Furthermore, we may write
\begin{equation}\label{eq:m1_and_m2_expressions}
 m_1=\frac{v}{k}\cos\varphi, \quad m_2=\frac{v}{k}\sin\varphi, \quad
m_1^2+m_2^2=\frac{v^2}{k^2}.
\end{equation}
As a result of the last relation, the equality $v=0$ occurs only at the poles. 
The solutions constructed in this paper assume the poles as values only in the limit $\xi\to\pm\infty$.
Thus, there is no conflict with Eq.~\eqref{eq:coordinateSystem} where $v$ appears in the denominator.

{\it Projection of Eq.~\eqref{eq:LLtravelingExplicit} on $\magn'/\vabs$.}
We obtain 
\begin{equation}  \label{eq:projection_mprime2}
 \vel \vabs^2=\magn\cdot(\magn''\times \magn')+k^2 m_3\ez\cdot(\magn'\times\magn)+ 2\dm\vabs^2 m_1. 
\end{equation}

It is advantageous at this point to define the unit tangent vectors $\ev$ and $\eth$  in the direction of increasing $\varphi$ and $\theta$ respectively,
\begin{equation}
 \ev=(-\sin\varphi, \cos\varphi, 0), \quad
\eth=(\cos\theta\cos\varphi, \cos\theta\sin\varphi, -\sin\theta),
\end{equation}
and to let 
\begin{equation}  \label{eq:linear_velocities}
  \magn'= w\eth + u\ev.
\end{equation}
The ``linear velocity'' components $w$ and $u$ are, respectively, along a parallel with radius equal to $\sin\theta$ (counter-clockwise) and a meridian (from north to south).
They  are given by 
\begin{equation}  \label{eq:wu}
 w=\theta',\quad u=\varphi'\sin\theta.
\end{equation}
Note the relation $|\magn'|=\sqrt{w^2+u^2}=v$.

The second term on the right of Eq.~\eqref{eq:projection_mprime2} is calculated using the expressions of $m_1$, $m_2$ from Eqs.~\eqref{eq:m1_and_m2_expressions}, as well as Eqs.~\eqref{eq:conservationLaw} and \eqref{eq:wu},
\begin{equation}  \label{eq:projection_mprime3}
k^2 m_3\ez\cdot(\magn'\times\magn) = -\vabs^2\varphi' \cos\theta = -ku\vabs\cos\theta.
\end{equation}
In order to simplify the first term on the right of \eqref{eq:projection_mprime2} we differentiate both sides of Eq.~\eqref{eq:linear_velocities} and obtain
\begin{align}
 \magn'' & = w\eth'+u\ev' + w'\eth+u'\ev \notag \\
  & = w( \thth \theta'+\thv \varphi') + u(\cancel{\vth\theta'}+\vv\varphi') + w'\eth+u'\ev,  
\end{align}
where the variable after the comma indicates the partial derivative with respect to that variable.
The term $\vth$ equals zero and the term $\thth$ is parallel to $\magn$ and does not contribute to the triple product $\magn\cdot(\magn''\times\magn')$.
The remaining vector derivatives are given by $\thv=(\cos\theta)\ev, \; \vv = -(\cos\varphi, \sin\varphi, 0).$
The contributing terms of $\magn''$ to be inserted in Eq.~\eqref{eq:projection_mprime2} are  
\begin{align}
  & (w\varphi' \cos\theta+u')\ev+u\varphi'\vv + w'\eth \notag  \\
  & \qquad\qquad = \left(\frac{kuw\cos\theta}{v} + u'\right)\ev + \frac{ku^2}{\vabs}\vv + w'\eth
\end{align}
where we have inserted $\varphi'=ku/\vabs$ from Eqs.~\eqref{eq:conservationLaw}, \eqref{eq:wu}.
In the cross product $\magn''\times\magn'$, the term $\vv\times\eth$ is orthogonal to $\magn$, thus, has vanishing projection on $\magn$.
We are left with the two combinations $\vv\times\ev=-\ez,\; \eth\times\ev=\magn$, and we obtain
\begin{equation}  \label{eq:projection_mprime4}
\magn\cdot(\magn''\times\magn') = w'u-u'w-ku \vabs\cos\theta.
\end{equation}
Inserting the results of Eqs.~\eqref{eq:projection_mprime3} and \eqref{eq:projection_mprime4} into
Eq.~\eqref{eq:projection_mprime2}, we finally obtain
\begin{equation}  \label{eq:projection_mprime5}
 \vel\vabs^2=w'u-u'w-2ku\vabs\cos\theta + 2\dm \vabs^2 \sin\theta \cos\varphi.
\end{equation}
Let us introduce a new variable $\rho$,
\begin{equation}  \label{eq:rho}
u = \vabs\sin\rho,\quad w = -\vabs\cos\rho
\end{equation}
such that $\rho=0$ corresponds to motion from the south to the north pole.
We insert Eqs.~\eqref{eq:rho} in Eq.~\eqref{eq:projection_mprime5} to obtain
\begin{equation}\label{eq:projection_mprime6}
 \vel \vabs^2=\vabs^2\rho'-2k\vabs^2\sin\rho\cos\theta + 2\dm\vabs^2 \sin\theta \cos\varphi.
\end{equation}
Two equations connecting $\rho$ with $\theta$ and $\varphi$ are derived using Eqs.~\eqref{eq:wu}, \eqref{eq:rho} and \eqref{eq:conservationLaw},
\begin{equation}  \label{eq:thetaprime}
 \theta' = - k\sin\theta\cos\rho,\quad \varphi' = k\sin\rho.
\end{equation}
Eqs.~\eqref{eq:projection_mprime6} and \eqref{eq:thetaprime} have the equilibrium solutions $\theta=0, \pi$ that correspond to $\vabs=0$ via Eq.~\eqref{eq:conservationLaw}.
We obtain the autonomous system of ordinary differential equations (ODEs) 
\begin{equation}  \label{eq:rho_theta_phi_system}
\begin{split} 
 \rho' & = 2k\cos\theta\sin\rho - 2\dm \sin\theta \cos\varphi + \vel,  \\
  \theta' & = -k\sin\theta\cos\rho, \\
  \varphi' & = k \sin\rho.
\end{split}
\end{equation}

\section{Existence of traveling domain walls}
\label{sec:existence}

\subsection{The system of equations for traveling chiral domain walls}

The presence of the DM interaction ($\dm\neq 0$) in Eqs.~\eqref{eq:rho_theta_phi_system} allows us to prove the existence of traveling domain wall solutions, that connect the two  equilibrium points  $\theta=0$ and $\theta=\pi$, representing constant magnetization at the north and the south pole of the sphere respectively.  The poles are reached only in the limit  of the variable $|\xi|$ tending  to infinity, the fact that the variable $\varphi$ is not defined at the poles is not a hindrance. 

We make a change of variable that transforms the sphere to a cylinder, allowing for simpler calculation.
We define a new variable $z$ by $\tan \left(\textstyle{\frac{1}{2}}\theta\right)=e^{-z}$. 
Making use of the formulae 
\begin{equation}
 \cos\theta=\tanh z, \quad
\sin\theta=\sech z, \quad  d\theta=-(\sin\theta)dz,
\end{equation}
the system of Eqs.~\eqref{eq:rho_theta_phi_system} becomes
\begin{subequations}  \label{eq:rho_z_phi_system}
\begin{align} 
\rho' & = 2k\tanh z\sin\rho -2\lambda \sech z \cos\varphi+\vel \label{eq:rho_z_phi_system1} \\
z' & = k\cos\rho \label{eq:rho_z_phi_system2} \\
\varphi' & = k\sin\rho. \label{eq:rho_z_phi_system3}
\end{align}
\end{subequations}
The gradient of the right side of these equations with respect to the  three dependent variables remains globally bounded.
Therefore the ODE system exhibits global existence and uniqueness of solutions as well as continuous dependence of the solution on their initial data and on the system parameters.
We have the obvious relation 
\begin{equation}  \label{eq:circle_equation}
 z'^2+\varphi'^2=k^2. 
\end{equation}

The circular cylinder (of radius unity) formed in the new coordinates when the points $(z,\varphi)$ and $(z,\varphi+2\pi)$ are identified, is topologically equivalent to the  original sphere $(\theta,\varphi)$, punctured at its north and the south poles. The north pole of the sphere corresponds to $z=+\infty$, the south pole  to $z=-\infty$.
The meridians of the sphere map to the lines $\varphi=$\,constant, that are parallel to the axis of the cylinder (where the minimal principal curvature is equal to zero).
The parallels of the sphere map  to the intersection circles of the cylinder cut by  planes that are perpendicular to its axis. 

In the absence of the DM term ($\dm=0$), the system \eqref{eq:rho_z_phi_system} yields a standing domain-wall ($\vel=0$ and, thus, $\xi=x$).
The solution is $\sin\rho=0$, $\varphi=\mbox{constant}$ and $z=\pm kx$.
Taking the hyperbolic tangent of both sides in the latter relation and recalling that $m_3=\cos\theta=\tanh z$, obtains the standing domain walls respectively for $\rho=0$ and $\rho=\pi$,
\begin{equation}  \label{eq:DWstandard}
m_3=\pm\tanh(kx)
\end{equation}
The plus sign corresponds to $\rho=0$ and the minus sign corresponds to $\rho=\pi$.
The domain-wall \eqref{eq:DWstandard} survives even in the presence of the DM term for $\varphi(x)=\pm\frac{\pi}{2}$, which eliminates the DM term in Eqs.~\eqref{eq:rho_z_phi_system}.

\subsection{Symmetry and choices}
We require the functions $z(\xi)$ and $\rho(\xi)$  to be  odd  and the function $\varphi(\xi)$ to be even.
This is consistent with the fact that when the sign of $\xi$ is reversed the substitution 
\begin{equation}  \label{eq:symmetries}
 (\rho,z,\varphi)\mapsto(-\rho,-z,\varphi),
\end{equation}
leaves the system of Eqs~\eqref{eq:rho_z_phi_system} invariant. 
Effectively, this  restricts the analysis of the system to the domain $\xi\ge0$ adopting  the initial data 
\begin{equation}  \label{eq:initial_data}
  \rho(0)=0, \quad z(0)=0, \quad \varphi(0)=\varphi_0.
\end{equation}
The condition that $\rho=0$ at the wall center ($z=0$) means that the domain wall goes from the south to the north pole, \ie, $m_3\to \mp \infty$ as $\xi\to\pm\infty$.
More general domain walls can be easily obtained as explained in Sec.~\ref{subsec:moreWalls}.
We narrow our search  to {\it strictly monotone} traveling domain walls.
As a domain wall connects the two poles $z=\pm\infty$ and as Eq.~\eqref{eq:rho_z_phi_system2} implies $z'(0)=k>0$
(we have tacitly assumed that $k>0$),
monotonicity means that $z(\xi)$ is strictly increasing and $z\to+\infty$, in the limit $\xi\to+\infty$. 

Adopting the further conditions 
\begin{equation}  \label{eq:dm_phi0}
  \dm>0, \qquad  -\frac{\pi}{2}\le \varphi_0\le \frac{\pi}{2}, 
\end{equation}
guarantees that, for the velocity $\vel>0$, the DM term counteracts the velocity term in Eq.~\eqref{eq:rho_z_phi_system1} as $\xi$ increases from its zero value.
From the same equation it is clear that there must hold $\vel < 2\dm$ in order to have a monotone domain wall.
This velocity bound has been also derived in Eq.~\eqref{eq:speedBound1}.

\begin{theorem}  \label{thm:asymptotics}
A strictly monotone domain wall exhibiting symmetries \eqref{eq:symmetries} and taking the values \eqref{eq:initial_data} at its center, satisfies the relations
\begin{subequations}\label{rho_asymptotic}
\begin{align}  \label{eq:rho_asymptotic1}
&\begin{cases}\lim_{z\to+\infty}\sin\rho=-\frac{\vel}{2k} \\
  k\sin\rho=-\frac{\vel}{2}- \frac{\vel}{2} e^{-2z}+2\dm b(z,\vel,\dm)e^{-z}, \quad |b|<2,
 \end{cases}\\
 \label{eq:rho_asymptotic2}
&\lim_{\xi\to+\infty}z'(\xi)=\sqrt{k^2-(\textstyle{\frac{\vel}{2}})^2},\\
\label{eq:rho_asymptotic3}
&\lim_{\xi\to+\infty}\varphi'(\xi)=-\frac{\vel}{2}.
\end{align}
\end{subequations}
In all limits, the convergence is exponential.
The magnetization vector converges to the north pole as $\xi$ increases ($z'>0$) and the domain wall travels to the right.
\end{theorem}

\begin{proof} 
 Since $z(\xi)$ is strictly increasing, we combine  Eqs.~\eqref{eq:rho_z_phi_system1} and \eqref{eq:rho_z_phi_system2} to make a change of the independent variable from $\xi$ to $z$, 
\begin{equation}  \label{eq:sine_rho}
k\frac{d}{dz}\sin\rho-2k\tanh z\sin\rho
=\vel-2\dm \sech z\cos\varphi. 
\end{equation}
 Multiplying both sides of the equation by the integrating factor $\sech^2z$ obtains
\begin{equation}  \label{eq:sine_rho2}
\frac{d}{dz}\left(k\sech^2z\sin\rho\right)=\vel\,\sech^2 z-2\dm \sech^3 z\cos\varphi. 
\end{equation}
 Integrating from an arbitrary $z$ to positive infinity, multiplying  both sides  by $\cosh^{2} z$ and doing simple algebra obtains
\begin{equation} \label{eq:sine_rho_equation2}
k\sin\rho=-\frac{\vel}{1+\tanh z}+2\dm\cosh^2z\int_z^{+\infty} \sech^3  \eta\cos\varphi(\eta)d\eta.
\end{equation}
We note that
\[
\left| \cosh^2z\int_z^{+\infty} \sech^3\eta\cos\varphi(\eta)\,d\eta \right|
< \cosh^2z\int_z^{+\infty} \sech^3\eta\, d\eta
< 2\,e^{-z}.
\]
We now obtain Eq.~\eqref{eq:rho_asymptotic1},
where
\[
b(z,\vel,\dm) = e^{-z}\, \cosh^2z\int_z^{+\infty} \sech^3  \eta\cos\varphi(\eta)d\eta.
\]
Inserting the limit of Eq.~\eqref{eq:rho_asymptotic1} in Eqs.~\eqref{eq:rho_z_phi_system2} and \eqref{eq:rho_z_phi_system3} we obtain
Eqs.~\eqref{eq:rho_asymptotic2} and \eqref{eq:rho_asymptotic3} respectively.
\end{proof}

If we integrate Eq.~\eqref{eq:sine_rho2} from zero to $z$, we obtain
\begin{equation}  \label{eq:sine_rho_eq1}
k\sin\rho=\frac{\vel}{2}\sinh 2z-2\dm\cosh^2z\int_0^z \sech^3 \eta\cos\varphi(\eta)d\eta.
\end{equation}
The right side of this equation represents a  balance between two terms that tend  to infinity as $z$ increases.
The balance is delicate; controlled by the angle $\varphi$, it forces the right side to stay between $\pm k$.
The rate of change of $\varphi$, 
\begin{equation}  \label{eq:phi_equation1}
\frac{d\varphi}{dz} = \tan\rho,
\end{equation}
is obtained by combining Eqs.~\eqref{eq:rho_z_phi_system2} and \eqref{eq:rho_z_phi_system3}.
These formulae are valid as long as $z(\xi)$ is monotone. The boundary of  monotonicity is reached at $\sin\rho=\pm1$. 

{\it Remark:} One may need the integral of $\sech^3z$, {\it e.g.} for the purpose of producing bounds. The integral is computable by exact formula
\begin{subequations}
\begin{align}
  \int\sech^3z\,dz=&\frac{1}{2}\tanh z\sech z
 +\frac{1}{2}\arctan(\sinh z)+\mbox{const.},   \\
 \int_z^\infty\sech^3 \eta\, d\eta=&-\frac{1}{2}\tanh z\sech z
 +\left(\frac{\pi}{2}
 -\frac{1}{2}\arctan(\sinh z)\right),  \\
 J=& \int_0^{\infty}\sech^3\eta\, d\eta=\frac{\pi}{4}.
\end{align}
\end{subequations}

\subsection{Existence of domain-wall solutions: A topological approach}

We begin by proving the following technical theorem.

\begin{lemma}  \label{thm:technical}
Let $\loverk=\dm/k$ and let the function $f$ be defined by the formula
\begin{equation}
 f(\loverk) = \loverk(\sinh 2z_*+2S(z_*)\cosh^2z_*),
\end{equation}
where
\begin{equation}
z_* =\ln\left(\textstyle\loverk+\sqrt{1+\left(\loverk\right)^2}\right)
-\ln\left(1-\textstyle\loverk\right), \qquad S(z)=\int_0^z\sech^3(\eta)d(\eta).
\end{equation}
Let the  ordering \eqref{eq:parameter_ordering} of the parameters $k,\, \dm, \,\vel$ be satisfied.
Then, for all $\loverk$ for which $f(\loverk)<1$, the orbit with initial conditions \eqref{eq:initial_data} either never reaches the boundary $\rho=\pm\frac{\pi}{2}$  
or it reaches the boundary  transversely ($\rho'\ne0$). 
\end{lemma}
\begin{proof}
We prove the theorem with the aid of two claims

{\it Claim 1. If $z_*\le
z$, then the equalities $\cos\rho=0$ and $\rho'=0$ cannot hold simultaneously.} 

Proof of claim 1. Inserting $\cos\rho=0$ and $\rho'=0$ in the first Eq.~\eqref{eq:rho_z_phi_system}, we obtain 
\begin{equation}
 0=\pm 2k\tanh z - 2\dm \sech z \cos\varphi+\vel,
\end{equation}
which is expressed as a quadratic equation in $e^z$,
\begin{equation}
 \left(\pm k+\frac{\vel}{2}\right) e^{2z} - 2\dm e^{z} \cos\varphi- \left(\pm k-\frac{\vel}{2}\right)=0.
\end{equation}
The roots  are
\begin{equation}
 e^z=\frac{\lambda\cos\varphi\pm\sqrt{\lambda^2\cos^2\varphi+k^2-\left(\frac{\vel}{2}\right)^2}}{\pm k+\frac{\vel}{2}}.
\end{equation}
Considering the constraints \eqref{eq:parameter_ordering} the following estimate applies,
\begin{equation}
e^z\le
 \left(\textstyle\frac{\dm}{k}+\sqrt{1+\left(\frac{\dm}{k}\right)^2}\right)
\left(1-\textstyle\frac{\dm}{k}\right)^{-1}. 
 \end{equation}
Taking the logarithm on both sides, produces  exactly the inequality $z\le z_*$. 
The exponent $z_*>0$ can be made arbitrarily small by having values of $\dm/\anisotropy > 0$ sufficiently small.  

With claim 1 now proved, we still need to exclude the simultaneous holding of $\cos\rho=0$ and $\rho'=0$. Claim 2 does more than this.

{\it Claim 2. $\cos\rho\ne 0$ when $z\le z_*$}. 
 
Proof of Claim 2. When  $z\le z_*$, we obtain easily that $kf(\loverk)$ is absolutely greater than the right hand side of Eq.~\eqref{eq:sine_rho_eq1} and hence than the left side, $k\sin\rho$.
The requirement  $f(\loverk)<1$ of the theorem, implies $|\sin\rho|<1$, hence $\cos\rho\ne0$.
\end{proof}

\begin{theorem}  \label{thm:existence}
[Local existence of traveling domain walls]
Let $\loverk=\dm/k$ satisfy the condition of Lemma~\ref{thm:technical}.
Then, there is a neighborhood $N$ of $\frac{\pi}{2}$, such that for $\varphi_0\in N$ there is a strictly increasing domain wall solution of the system of Eqs.~\eqref{eq:rho_z_phi_system}.
The domain wall velocity is $\vel>0$ if $\varphi_0<\frac{\pi}{2}$ and $\vel<0$ if $\varphi_0>\frac{\pi}{2}$.
\end{theorem}

\begin{proof}
We fix the value $\loverk=\dm/k$ so that it satisfies the condition of Lemma~\ref{thm:technical}.
We assume that  the initial value of the angle coordinate $\varphi_0$ is less than $\frac{\pi}{2}$.
The case $\varphi_0>\frac{\pi}{2}$ will be proved by symmetry. 
We then consider the solution trajectories of the system of Eqs.~\eqref{eq:rho_z_phi_system} with initial values \eqref{eq:initial_data} and with velocity $\vel$ in a closed interval $[0, \vel_+]$, where $\vel_+=2\dm$.
We partition this velocity interval into the following three subsets.
\begin{enumerate}
\item Subset A: contains the  values of $\vel$ for which the first occurrence of $\cos\rho=0$ is at $\rho=\frac{\pi}{2}$. 
\item Subset B: contains the values of $\vel$ for which the first occurrence of $\cos\rho=0$ is at $\rho=-\frac{\pi}{2}$ 
\item Subset C: contains the remaining points of the  set $[0,\vel_+]$, that is all the velocities $\vel$ for which $-\frac{\pi}{2}<\rho<\frac{\pi}{2}$ for all $\xi\ge 0$. 
\end{enumerate}
The sets A and B are open in $[0, \vel_+]$, due to the continuous dependence of the orbits of  system~\eqref{eq:rho_z_phi_system} on the velocity $\vel$ and due to the fact that orbits reaching the boundary planes $\sin\rho=\pm\frac{\pi}{2}$ in phase space do so transversely, as shown in Lemma~\ref{thm:technical}.
In the claims below we show that the velocity $\vel=\vel_+$ and, hence, a neighborhood of it  belongs to the set $A$, the velocity $\vel=0$ and a neighborhood of it  belongs to the set $B$.
Thus, the sets A and B are nonempty.

The sets A and B are disjoint and open in the closed interval $[0,\vel_+]$.
Therefore, the complement C of their union is nonempty.
As a result, there is a velocity $\vel$ with a trajectory for which $\sin\rho\in(-1,1)$ for $z>0$ and thus $z(\xi)$ is a monotone function that converges to infinity.
This solution represents a domain wall. 
 
\noindent{\it Claim 1. The set A is nonempty.} We consider the orbit with $c=c_+=2\dm$. As a result, the right side of Eq.~\eqref{eq:rho_z_phi_system1} remains positive, when $z>0$, and $\rho$ converges to positive infinity due to its first term. Thus, the orbit crosses the value $\rho=\frac{\pi}{2}$.

\noindent{\it Claim 2. The set B is nonempty.} We consider the orbit with $\vel=0$.
As the point $(\rho,\,z,\varphi)$ evolves from $(0,\,0,\,\varphi_0)$, we will prove that, for suitable values of $\loverk=\dm/k$ and $\varphi_0$, the variable $\rho$ will reach the value $-\frac{\pi}{2}$ monotonically;
this orbit will then belong to set B, thus proving that it is nonempty.
Our proof considers small positive values of $\frac{\pi}{2}-\varphi_0$.
Nevertheless, it allows one to see how a traveling domain wall arises continuously from the static domain wall, at which $\varphi=\frac{\pi}{2}$.
Numerical experiments in Sec.~\ref{sec:numerics} demonstrate the validity of the claim for a much wider range of parameters.

Making the following simplifying changes of notation and introducing the function $g(z)$,
\begin{equation} \label{def:variables}
q=-\sin\rho, \qquad \pold=\cos\varphi, 
\qquad g(z)=\frac{\sqrt{1-\pold^2}}{\sqrt{1-q^2}}, 
\end{equation}
we rewrite Eqs.~\eqref{eq:sine_rho} and~\eqref{eq:phi_equation1}, respectively, as the system
\begin{subequations}  \label{eq:q_psi_system}
\begin{align}
\label{eq:q_psi_system1}
\dfrac{dq}{dz} =&\, 2q\tanh z +2\loverk \pold\sech z \\ 
 \dfrac{d\pold}{dz} =&\, qg.
\end{align}
\end{subequations}
and we impose initial conditions $q(0)=0$ and $\pold(0)=\pold_0$.
As $\pold_0\to 0$ the solutions of the linearized system 
\begin{equation}  \label{eq:q_psi_system_approx}
\begin{split}
\dfrac{dq}{dz} =&\, 2q\tanh z +2\loverk \pold\sech z \\ 
 \dfrac{d\pold}{dz} =&\, q
\end{split}
\end{equation}
approximate the solutions of the full system \eqref{eq:q_psi_system} uniformly in compact sets $[0,\,z_1]$.
The substitution $q=\nu\pold$ transforms the system~\eqref{eq:q_psi_system_approx} to 
\begin{equation}  \label{eq:nu_psi_system}
\begin{split}
\dot\nu=&-(\nu^2-2\nu\tanh z-2\loverk\sech z)=-(\nu-\nu_+)(\nu-\nu_-), \qquad \nu(0)=0,  \\ 
 \dot\pold=&\nu\pold, \qquad \pold(0)=\pold_0,
\end{split}
\end{equation}
where dots on top indicate derivatives taken with respect to $z$ and 
\begin{equation}
 \nu_\pm=\textstyle{\frac{1}{2}}(\tanh z\pm\sqrt{\tanh^2 z+2\loverk\sech^2 z}). 
\end{equation}
We make the following observations: (i) $\dot\nu$ is negative above the graph of $\nu_+(z)$ and below the graph of $\nu_-(z)$; it is positive between the two graphs, (ii) $\nu(0)$ lies between the two graphs, and (iii) $\nu_+(z)\to1$ monotonically as $z$ increases. As a result,  $\nu(z)\to 1$ as $z$ increases.

Returning to the original variables, we have 
\begin{equation}
\pold(z)=\pold_0e^{\int_0^z\nu(\eta)d\eta},  \qquad q(z)=\nu(z)\pold(z)=\nu(z)\pold_0e^{\int_0^z\nu(\eta)d\eta}. 
\end{equation}
For the linear approximation to be valid, we need $\pold_0\ll e^{-\int_0^z\nu(\eta)d\eta}$.
In this way a large $z$ is reached, while $\pold$ and $q$ are still small and nearly equal to each other ($\nu$ approaches $1$).
In Eq.~\eqref{eq:q_psi_system1}, the second term on the right becomes negligible, as $\sech z\to 0$, and the variable $q$ increases to $1$ monotonically due to the first term. This proves the claim. 

In order to prove the theorem, all we have to do is to select a $\pold_0$ sufficiently small. 
\end{proof}

\section{More traveling domain walls}

\subsection{Parity related domain walls} 
\label{subsec:moreWalls}

The domain walls whose existence has been proved in Theorem~\ref{thm:existence} go from the south to the north pole, \ie, $m_3\to \mp \infty$ as $\xi\to\pm\infty$, and the velocity is $\vel > 0$, \ie, they travel to the right.
For any such domain wall solution we can obtain further traveling domain walls with the following simple transformations.
These are easy to apply to Eqs.~\eqref{eq:LLtravelingExplicit} and corresponding transformation can be applied to Eqs.~\eqref{eq:rho_z_phi_system}.
\begin{enumerate}
\item
{\it South to north pole, positive velocity.}
We have obtained in Theorem~\ref{thm:existence} a domain wall $(m_1,m_2,m_3)$ that goes from the south to the north pole and has positive velocity $\vel>0$, \ie, it is traveling to the right.
\item
{\it South to north pole, negative velocity.}
Apply $\vel\to -\vel$ and $(m_1,m_2,m_3) \to (-m_1,m_2,m_3)$.
In terms of the variables used in Eqs.~\eqref{eq:rho_z_phi_system} we obtain the same transformation by $\varphi\to\pi-\varphi,\,\rho\to - \rho$.
\item
{\it North to south pole, positive velocity.}
Apply $(m_1,m_2,m_3) \to (m_1,-m_2,-m_3)$.
The same transformation is obtained by $z\to -z,\, \varphi\to -\varphi,\,\rho\to \pi+\rho$.
\item
{\it North to south pole, negative velocity.}
Apply $\vel\to -\vel$ and $(m_1,m_2,m_3) \to (-m_1,-m_2,-m_3)$.
The same transformation is obtained by $z\to -z,\, \varphi\to \pi+\varphi,\,\rho\to \pi-\rho$.
\end{enumerate}
For the last two cases the value $\rho$ at the wall center is $\rho(z=0)=\pi$, while we have assumed $\rho(z=0)=0$ in Sec.~\ref{sec:existence}.

Our results carry over to the case of the so-called interfacial DMI, in which the bulk DMI term in Eq.~\eqref{eq:heff} is replaced by the term $2\dm\ey\times\p_x\magn$.
Any solution $\magn$ of the original model is a solution of the interfacial DMI model if we rotate the magnetization vector by $-\pi/2$.

\subsection{Higher velocity domain walls}

If we obtain a traveling domain wall solution for specific parameter values $\anisotropy, \dm$, we can obtain further solutions by a straightforward scaling.
Let us assume $\magn_0(\xi)$ a traveling domain wall solution with velocity $\vel$ for specific values of the parameters $\anisotropy, \dm$ such that $\anisotropy/\dm>\pi/2$, so that \eqref{eq:parameter_ordering} is satisfied.
Then the configuration $\magn_\mu(\xi) = \magn_0(\xi/\mu)$, where $\mu$ is a constant, is a solution of Eq.~\eqref{eq:LLtravelingExplicit} under the transformation
\begin{equation}
\vel \to \mu\vel,\quad \anisotropy \to \mu \anisotropy,\quad \dm \to \mu\dm.
\end{equation}
This means that the maximum attainable domain wall velocity scales proportional to the anisotropy constant $\anisotropy$ if we keep the ratio $\anisotropy/\dm$ constant.
There is no theoretical limit to the velocity, and this is set only by the availability of materials with high anisotropy and DMI parameters.
Furthermore, the width of the domain wall scales inversely proportional to $\anisotropy$, which means that faster and narrower walls are obtained for increasing $\anisotropy$.
The above dynamical behavior is very different than that of the Walker domain walls, which become slower as $\anisotropy$ increases and a theoretical limit exists for $\anisotropy\to 0$ (see \ref{sec:WalkerWall}).

\section{Numerical calculation of propagating domain walls}
\label{sec:numerics}

We solve Eq.~\eqref{eq:LLtraveling} numerically by applying the relaxation algorithm \cite{KomineasPapanicolaou_PRB2015a}
\begin{equation}  \label{eq:relaxationAlgorithm}
\dot{\magn} = - \magn\times (\magn\times \heff - \vel \p_\xi\magn),\quad
\vel = u - \linmom.
\end{equation}
The velocity $\vel$ is determined self-consistently in terms of an arbitrary input parameter $u$ and the linear momentum for the one-dimensional system, defined as \cite{Haldane_PRL1986}
\begin{equation}  \label{eq:linearMomentum}
\linmom = \int \frac{m_1}{1-m_1^2} (m_2 m_3' - m_3 m_2')\, dx.
\end{equation}
If we note that $\magn\times\delta\linmom/\delta\magn = -\p_x\magn$ it becomes evident that Eq.~\eqref{eq:relaxationAlgorithm} is a minimization algorithm for the functional $\Energy+\frac{1}{2}(u-\linmom)^2$.
When the algorithm converges to a minimum of the functional, where $\dot{\magn}=0$, the magnetization configuration satisfies Eq.~\eqref{eq:LLtraveling} and represents a solitary wave with velocity $\vel$.

The form \eqref{eq:linearMomentum} for the definition of the linear momentum is not unique.
Among the possible definitions, we have chosen the one which is well-defined (i.e., the one that contains a non-divergent integrand) and takes the value $\linmom=0$ for the static Bloch wall \eqref{eq:BlochWall}.
One could add a total derivative in the integrand of Eq.~\eqref{eq:linearMomentum} that give a value proportional to $\pi$ upon integration for a domain wall.
The ambiguity in the value of the linear momentum does not have any effect on our calculations since Eq.~\eqref{eq:relaxationAlgorithm} is solved for any choice of the linear momentum definition.

For a definite numerical calculation we choose the anisotropy and DM parameters
\begin{equation}  \label{eq:params1}
\anisotropy = 1.0,\qquad
\dm = 0.5
\end{equation}
so that $\anisotropy > \anisotropy_c$ and thus the ground state is the (uniform) ferromagnetic.
We insert the Bloch wall \eqref{eq:BlochWall} (choosing the plus signs) as an initial configuration in the numerical algorithm \eqref{eq:relaxationAlgorithm}.
We vary the input parameter $u$ in a range of values and the algorithm converges to a domain wall profile.
The linear momentum $\linmom$ is calculated from the profile and we obtain a velocity range $0 \leq \vel < \velmax$ with $\velmax \approx 0.78$.  

\begin{figure}[t]
\centering
(a)
\includegraphics[width=0.45\linewidth]{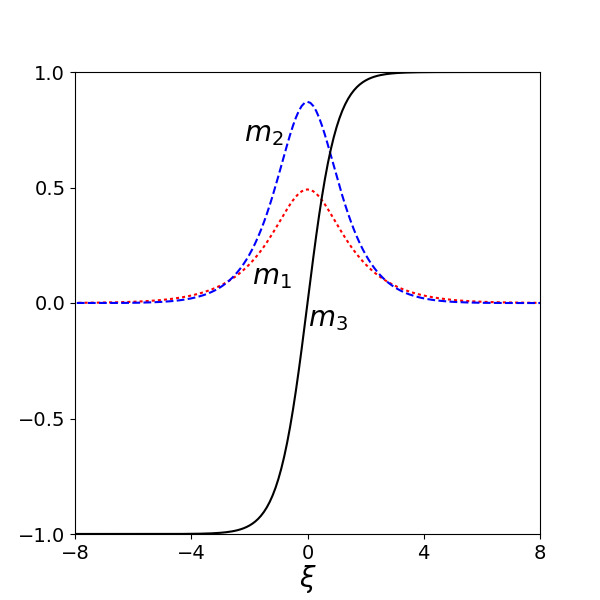}\hspace{15pt}
(b) \includegraphics[width=0.45\linewidth]{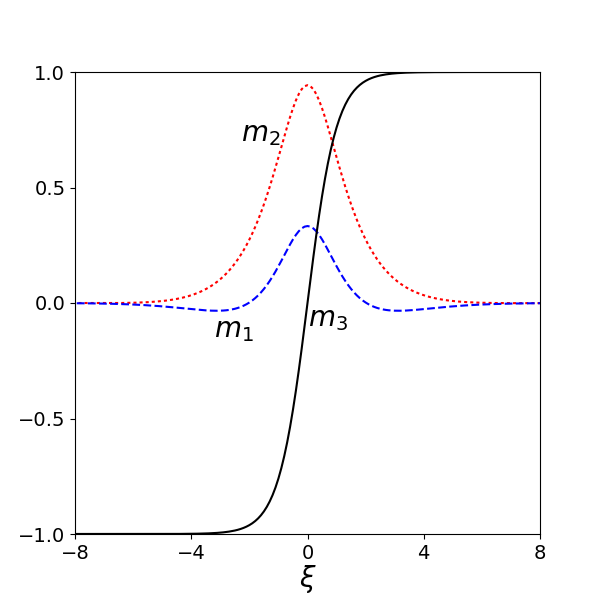}
\caption{The three components of the magnetization for domain walls with velocities (a0 $\vel=0.40$ and (b) $\vel=0.75$, traveling to the right.
These walls have a tilting so that $m_2(0) > 0$ (or $0 < \varphi_0 < \pi/2$), thus they belong to the branch of domain walls around the stable Bloch wall.
Parameter values are given in Eq.~\eqref{eq:params1}.}
\label{fig:DWprofiles}
\end{figure}

For $\vel=0$ we have $m_1=0$ and, as the velocity increases, the component $m_1$ increases, \ie, the magnetization tilts to the direction of motion.
In Fig.~\ref{fig:DWprofiles} we show propagating domain wall profiles for an intermediate velocity $\vel=0.4$ and for a high velocity $\vel=0.75$.
In the limit $\vel \to \velmax$ we have $m_1\to 1$ at the domain wall center.
In the same limit $m_2$ goes to zero at the wall center while it is small, but does not vanish, for $\xi\neq 0$.
The linear momentum \eqref{eq:linearMomentum} takes the value $\linmom=0$ for $\vel=0$ and it is found numerically to be $\linmom \approx 3.14$ when $\vel$ is close to $\velmax$.
The limit $\vel\to\velmax$ is not easy to approach numerically as the denominator in the definition \eqref{eq:linearMomentum} of the linear momentum takes very small values.

\begin{figure}[t]
\centering\includegraphics[width=0.45\linewidth]{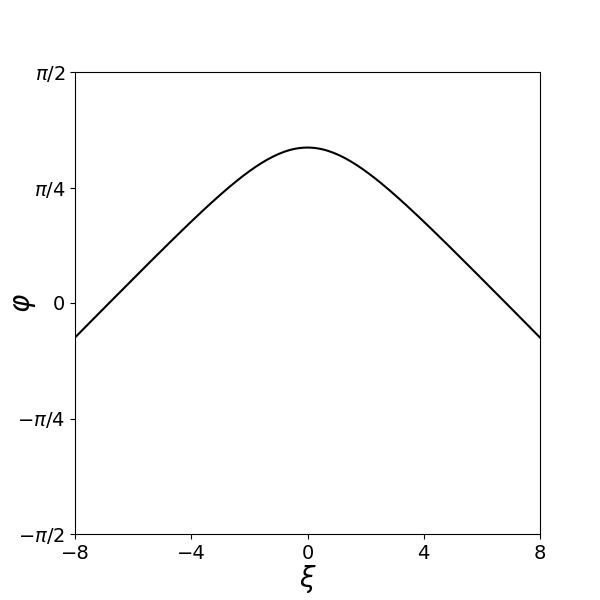}
\caption{The phase angle $\varphi(\xi)$ for the domain wall shown in Fig.~\ref{fig:DWprofiles}a with velocity $\vel=0.4$.
The asymptotic behavior of Theorem~\ref{thm:asymptotics} is followed.}
\label{fig:angle}
\end{figure}

\begin{figure}[t]
\centering
(a)
\includegraphics[width=0.45\linewidth]
{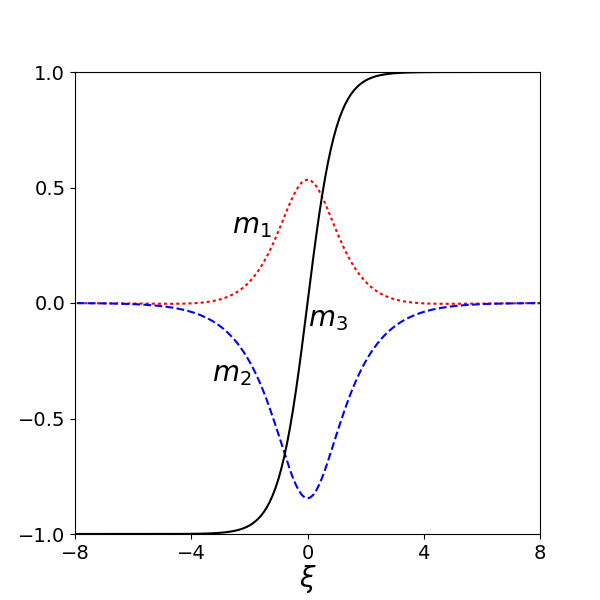}\hspace{15pt}
(b)
\includegraphics[width=0.45\linewidth]{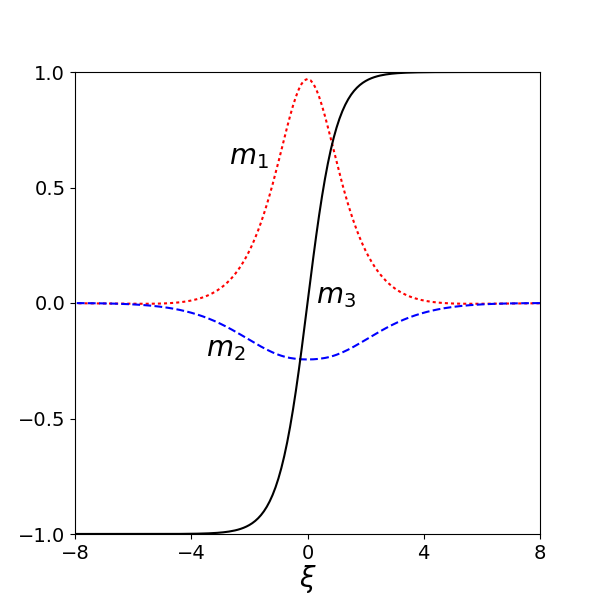}
\caption{The three components of the magnetization for domain walls with velocities (a) $\vel=0.40$ and (b) $\vel=0.75$, traveling to the right.
These walls have a tilting so that $m_2(0) < 0$ (or $-\pi/2 < \varphi_0<0$), thus they belong to the branch of domain walls around the unstable Bloch wall. 
Parameter values are given in Eq.~\eqref{eq:params1}.}
\label{fig:DWprofilesUnstable}
\end{figure}

The complex structure of the domain wall profiles is more apparent for the faster moving walls in Fig.~\ref{fig:DWprofiles}.
We have verified that our numerically calculated domain wall profiles verify the results stated in Theorem~\ref{thm:asymptotics}.
Fig.~\ref{fig:angle} shows the angle $\varphi(\xi)$ for the domain wall shown in Fig.~\ref{fig:DWprofiles} with velocity $\vel=0.4$.
The line follows the asymptotic behavior given in Theorem~\ref{thm:asymptotics}.
This indicates clearly an oscillating behavior of $m_1, m_2$ at the tails of the domain wall,
although this behavior is only slightly visible in the scale of Fig.~\ref{fig:DWprofiles}.
We have found so far a branch of domain walls around the stable Bloch wall.

\begin{figure}[t]
\begin{center}
(a)
\includegraphics[width=0.40\linewidth]{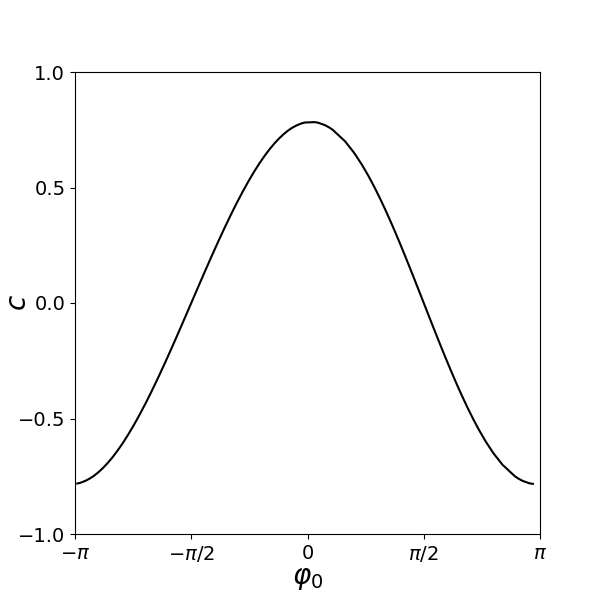}\hspace{20pt}
(b)
\includegraphics[width=0.40\linewidth]{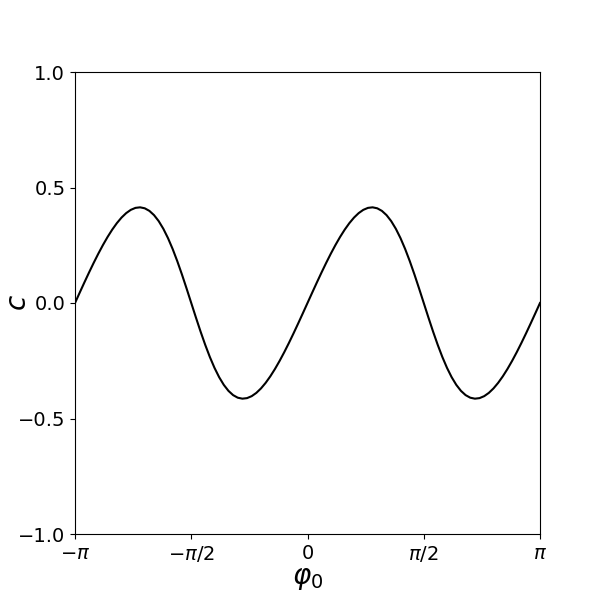}
\end{center}
\caption{(a) The velocity $\vel$ versus tilting angle $\varphi_0$ of the magnetization at the center of the domain wall, for a DM ferromagnet with parameter values as in Eq.~\eqref{eq:params1}.
Note that the domain walls corresponding to $\pm\varphi_0$ are not related by a simple symmetry transformation.
(b)
The velocity $\vel$ for the Walker wall given in Eq.~\eqref{eq:velocity-Walker}, for $\anisotropy=1$, versus tilting angle $\varphi_0$.}
\label{fig:vel_vs_phi}
\end{figure}

We can obtain further traveling domain walls by giving as initial condition to our numerical algorithm the Bloch wall profile \eqref{eq:BlochWall} where we choose the plus sign for $m_3$ and the minus sign for $m_2$.
This is a local energy {\it maximum} due to the contribution of the DMI and it it thus an unstable solution.
The algorithm converges for a velocity in the range $0 < \vel < \velmax$ with $\velmax$,  
the same one found earlier for the domain walls of Fig.~\ref{fig:DWprofiles}.
In Fig.~\ref{fig:DWprofilesUnstable} we show two of the propagating domain wall profiles for velocities $\vel=0.4$ and $\vel=0.75$.
These profiles are similar but not identical to the domain walls shown in Fig.~\ref{fig:DWprofiles}.
In other words, we cannot obtain this new set of domain wall profiles by simple transformations of the  profiles in Fig.~\ref{fig:DWprofiles}
The disparity is due to the chirality of the DMI.
On the other hand, the general features of all domain wall profiles are similar, in particular, they satisfy the asymptotic behavior given in Theorem~\ref{thm:asymptotics}.

For a parametrization of the domain walls we consider the tilting angle $\varphi_0$ of the magnetization at the center of the wall.
In Fig.~\ref{fig:vel_vs_phi}a we plot the velocity for the family of propagating domain walls versus the titling angle $\varphi_0$.
The maximum velocity is obtained for $\varphi_0=0$ (N\`eel wall).
The function $\vel=\vel(\varphi_0)$ is periodic with period $2\pi$. We note that the graph is not symmetric around $\varphi_0=0$ (although this is not apparent in the figure) as walls for $\varphi_0$ and $-\varphi_0$ are not related by a simple transformation.
The above should be contrasted to the Walker velocity in Eq.~\eqref{eq:velocity-Walker} that has a period of $\pi$ and possesses the symmetry $\vel\to-\vel$ for $\varphi_0\to-\varphi_0$.

\section*{Acknowledgement}

SK and CM gratefully acknowledge financial support by the DFG under the grant no. ME 2273/3-1 including a three-month Mercator fellowship during which this work was initiated.
SV gratefully acknowledges financial support by the NSF through contract DMS-1211638.
The authors thank the University of Crete and RWTH Aachen for their hospitality during reciprocal visits.

\appendix

\section{Walker wall}
\label{sec:WalkerWall}

Let us assume the model with symmetric exchange and easy-axis anisotropy where we add the magnetostatic interaction.
The latter is reduced to a relatively simple term, equivalent to modeling a hard axis in the film plane, if we assume that the film thickness is infinite.
The effective field for this well-studied model \cite{SchryerWalker_JAP1974,ODell} is
\begin{equation}  \label{eq:heff_Walker}
 \heff = \p_x^2\magn + \anisotropy^2\, m_3 \ez - m_1 \ex
\end{equation}
where the third term on the right hand side is an anisotropy of the easy-plane type which is modeling the magnetostatic field.
Distances are measured in exchange length units.

The LL equation \eqref{eq:LL} with effective field \eqref{eq:heff_Walker} has propagating domain wall solutions 
\begin{equation}
m_1 = \frac{\cos\varphi_0}{\cosh(\epsilon x)},\qquad 
m_2 = \frac{\sin\varphi_0}{\cosh(\epsilon x)},\qquad 
m_3 = \tanh(\epsilon x)
\end{equation}
under the following conditions
\begin{equation}  \label{eq:velocity-Walker}
\vel = -\frac{\sin(2\varphi_0)}{2\epsilon},\quad \epsilon = \pm \sqrt{\anisotropy^2+\cos^2\varphi_0}.
\end{equation}
These solutions follow a geodesic (meridian) on the magnetization sphere connecting the poles $m_3=\pm 1$.
The angle $\varphi_0$ is typically referred to as the {\it wall tilting}.
The Bloch wall of Eq.~\eqref{eq:BlochWall} is obtained for $\varphi_0=\pm\pi/2$, and the N\`eel wall is obtained for $\varphi_0=0, \pi$.
They are both static $(\vel=0$) solutions, but the Bloch wall is an energy minimum and is typically observed in experiments.

In Fig.~\ref{fig:vel_vs_phi}b we plot the velocity of the wall given in Eq.~\eqref{eq:velocity-Walker} as a function of the titling angle $\varphi_0$. 
The velocity is $\pi$-periodic due to the equivalence of the walls under the reflexion $\varphi_0 \to -\varphi_0$.
There is a maximum, that we shall call $\vel_{\rm max}$, obtained for a value of the wall tilting $\pi/4 < \varphi_0 < \pi/2$.
The tilting angle $\varphi_0$, and the corresponding maximum velocity $\vel_{\rm max}$, depend on the anisotropy $\anisotropy$.
We have $\vel_{\rm max} \to 1$ for $\anisotropy\to 0$, while $\vel_{\rm max} \approx 1/(2\anisotropy)$ for large values of $\anisotropy$.
To obtain the velocity in physical units we multiply the above dimensionless values by the unit of velocity $\vel_0=(\gamma_0\mu_0 M_s)\lex = \gamma_0 \sqrt{2\mu_0 A}$.
Thus, the largest possible value for the wall velocity is $\vel_0$ and it is obtained for small $\anisotropy$.

\section{Non-chiral magnet}
\label{sec:nonchiralMagnet}

The Landau-Lifshitz equation \eqref{eq:LL} where the effective field $\bm{f}$ contains only exchange and anisotropy terms has been shown to have a Lax pair and is completely integrable producing single and multi-soliton solutions as well as periodic solutions (see \cite{BikbaevBobenkoIts_TMP2014} and references therein).
Single soliton solutions without the use of integrability have been derived in \cite{TjonWright_PRB1977}.

Here, we are using the formulation in the main text of the present study and  derive single soliton and domain wall solutions without the use of integrability.
In the absence of DMI ($\dm=0$) the third equation in \eqref{eq:rho_theta_phi_system} decouples and, thus, we only need to solve the reduced system
\begin{equation}  \label{eq:gamma_theta_system}
 \begin{split} 
 \rho' & = 2k\sin\rho\cos\theta + \vel,  \\
  \theta' & = -k\sin\theta\cos\rho.
 \end{split}
\end{equation}
Cross-multiplying the two equations and further multiplying the result by $\sin\theta$ gives a perfect derivative
\begin{equation}
(\sin\rho\sin^2\theta)'-\frac{\vel}{k}(\cos\theta)' = 0.
\end{equation}
We have the integral
\begin{equation}  \label{eq:integral-nonchiral}
\alpha = \sin\rho\sin^2\theta -\frac{\vel}{k}\cos\theta.
\end{equation}

Inserting in Eq.~\eqref{eq:integral-nonchiral} $\cos\theta=m_3$ we obtain
\begin{equation}  \label{eq:cosrho-nonchiral}
(1-m_3^2)\sin\rho=\alpha+\frac{\vel}{k}m_3, \quad
\cos\rho=\pm\sqrt{1-\left(\frac{\alpha+\frac{\vel}{k} m_3}{1-m_3^2}\right)^2}.
\end{equation}
Multiplying both sides of the second equation by $-k\sin^2\theta$ and using the second of Eqs.~\eqref{eq:gamma_theta_system} gives $(\cos\theta)'$ on the left side, which equals $m_3'$. Thus, 
\begin{equation}  \label{eq:dm3dxi}
 \frac{d m_3}{d\xi}=\pm k(1-m_3^2)\sqrt{1-\left(\frac{\alpha+\frac{\vel}{k} m_3}{1-m_3^2}\right)^2}
= \pm \sqrt{-p(m_3)}
\end{equation}
where we have defined 
\begin{equation}  \label{eq:pm3}
  p(m_3) = - k^2 \left[ (1-m_3^2)^2 - \left(\alpha+\frac{\vel}{k} m_3 \right)^2 \right]. 
\end{equation}

For {\it domain walls} we have $\theta=0$ and $\theta=\pi$ at spatial infinity.
If we assume that $\theta\to 0$ at spatial infinity then Eq.~\eqref{eq:integral-nonchiral} gives $\alpha=-\vel/k$, while for $\theta\to \pi$ we have $\alpha=\vel/k$.
An immediate conclusion is that, for a domain wall where $\theta$ takes both values $0$ and $\pi$ at spatial infinity, we necessarily have to set $\vel=0$.

We continue our investigation anticipating {\it soliton solutions} with the asymptotic value $\theta\to 0$.
Inserting $\alpha=-\vel/\anisotropy$ in Eq.~\eqref{eq:pm3} we obtain
\begin{equation}
p(m_3) = -\anisotropy^2 (1-m_3)^2 \left(1+\frac{\vel}{\anisotropy}+m_3 \right)\left(1-\frac{\vel}{\anisotropy}+m_3 \right).
\end{equation}
Soliton solutions with far-field $m_3=1$ are obtained since $p(m_3)$ has a double root at $m_3=1$ and a neighboring single root.
For $\vel \geq 0$ the single root is at $m_3=-1+\vel/\anisotropy$, which allows for a soliton velocity in the range $0 \leq \vel < 2k$.
For $\vel \leq 0$ the single root is at $m_3=-1-\vel/k$, which allows for a soliton velocity in the range $-2k < \vel \leq 0$.
Thus, we have solitons for velocity $|\vel| < 2k$.
Alternatively, inserting $\alpha=\vel/\anisotropy$ in Eq.~\eqref{eq:pm3} we obtain similar results for solitons with far-field $m_3=-1$.

In the general case, the polynomial $p(m_3)$ factors to 
\begin{equation}
p(m_3) = \anisotropy^2 \left( m_3^2-\frac{\vel}{k}m_3-1-\alpha \right)\left( m_3^2+\frac{\vel}{\anisotropy}m_3-1+\alpha \right)
\end{equation}
For values of the constant $-1 < \alpha < 1$ the polynomial $p(m_3)$ has exactly one root in the interval $(0,1)$ and another root in $(-1,0)$.
In order to see this one calculates that $p(m_3=\pm 1) > 0$ and $p(m_3=0) < 0$.
Necessarily there are two roots in $(-1,1)$ and two roots outside it. 
Setting $m_3$ equal to one of the roots of $p(m_3)$ we obtain {\it spin wave solutions}.
These have $m_3$ constant according to Eq.~\eqref{eq:dm3dxi}, $\rho=\pm \pi/2$ from Eq.~\eqref{eq:gamma_theta_system} while for the angle $\varphi$ we have $\varphi'=\anisotropy$ from Eq.~\eqref{eq:rho_z_phi_system3}.
This gives precession of the magnetization vector with $\xi$.
For the first equation in \eqref{eq:gamma_theta_system} we find $m_3=\cos\theta=-\vel/2\anisotropy$.
This gives a bound for the velocity of spin waves $|\vel| <\vel/\anisotropy$.
Finally, there are also solutions in which $m_3$ oscillates between the two roots of $p(m_3)$.



\section*{References}

\end{document}